\documentclass[journal,12pt,onecolumn,draftclsnofoot,]{IEEEtran}
\usepackage{cite}
\usepackage{subfigure}
\usepackage{float}
\usepackage{lipsum}
\usepackage[utf8]{inputenc}
\usepackage{color}
\usepackage{graphicx}
\usepackage{amssymb}
\usepackage{multicol}
\usepackage{amsmath}
\usepackage{amsthm}
\usepackage[cmintegrals]{newtxmath}
\usepackage[normalem]{ulem}
\usepackage{mathtools}

\newtheorem{pro}{\textit{Proposition}}
\newtheorem{rem}{\textit{Remark}}

\definecolor{dblue}{rgb}{0,0,0.8}

\newcommand\blfootnote[1]{%
  \begingroup
  \renewcommand\thefootnote{}\footnote{#1}%
  \addtocounter{footnote}{-1}%
  \endgroup}

\begin{document}
\title{Spectral Efficiency of \\ One-Bit Sigma-Delta Massive MIMO
}
\author{Hessam Pirzadeh, \emph{Student Member, IEEE}, Gonzalo Seco-Granados, \emph{Senior Member, IEEE}, Shilpa Rao, \emph{Student Member, IEEE}, and A. Lee Swindlehurst, \emph{Fellow, IEEE}}

\maketitle
\begin{abstract}
We examine the uplink spectral efficiency of a massive MIMO base station employing a one-bit Sigma-Delta ($\Sigma\Delta$) sampling scheme implemented in the spatial rather than the temporal domain. Using spatial rather than temporal oversampling, and feedback of the quantization error between adjacent antennas, the method shapes the spatial spectrum of the quantization noise away from an angular sector where the signals of interest are assumed to lie. It is shown that, while a direct Bussgang analysis of the $\Sigma\Delta$ approach is not suitable, an alternative equivalent linear model can be formulated to facilitate an analysis of the system performance. The theoretical properties of the spatial quantization noise power spectrum are derived for the $\Sigma\Delta$ array, as well as an expression for the spectral efficiency of maximum ratio combining (MRC). Simulations verify the theoretical results and illustrate the significant performance gains offered by the $\Sigma\Delta$ approach for both MRC and zero-forcing receivers.
\end{abstract}

\blfootnote{This work was supported by the U.S. National Science Foundation under Grants CCF-1703635 and ECCS-1824565. The work of G. Seco-Granados has been partially funded by the Spanish Ministry of Science, Innovation and Universities under projects TEC2017-
89925-R and PRX18/00638 and by the ICREA Academia programme.

H. Pirzadeh, S. Rao, and A. L. Swindlehurst are with the Center for Pervasive
Communications and Computing, University of California, Irvine, CA 92697
USA (e-mail: hpirzade@uci.edu; shilpar1@uci.edu; swindle@uci.edu).

G. Seco-Granados is with the Department of Telecommunications and
Systems Engineering, Universitat Aut\`{o}noma de Barcelona, 08193 Barcelona,
Spain (e-mail: gonzalo.seco@uab.cat).}


\section{Introduction}\label{sec:Introduction}
\IEEEPARstart{T}{o} reduce complexity and energy consumption in large-scale MIMO systems, researchers and system designers have recently considered implementations with low-resolution analog-to-digital and digital-to-analog converters (ADCs, DACs).  Compared to hybrid analog/digital approaches, fully digital architectures, even with low-resolution sampling, provide increased flexibility and fully exploit the potentially large array gain promised by massive MIMO systems. The case of one-bit quantization has received the most attention, both for the uplink \cite{Fan,JZhang,juncil2015near,MollenCLH17,Li,JacobssonDCGS17,Dong,ShaoLL18,JeonLHH18} and downlink \cite{Nossek2009Transmit,SaxenaFS17,LiTSML17,CastanedaJDCGS17,Jacobsson2017,Swindlehurst2017,LandauL17,JeddaMSN18,li2018massive,Sohrabi2018,shao2018framework} scenarios. 

While one-bit ADCs and DACs offer the greatest simplicity and power savings, they also suffer the greatest performance loss compared to systems with higher resolution sampling, particularly for moderate to high signal-to-noise ratios (SNRs), and in situations with strong interference. Besides simply increasing the ADC/DAC resolution, mixed-ADC architectures \cite{Tan,JZhang2,Hessam_lett,Park} and temporal oversampling \cite{gok2017,Ucuncu2018,schluter2018,Galton_rev,Palguna} have been proposed to bridge the performance gap, with a corresponding increase in complexity and power consumption.

Oversampled one-bit quantization has a long history in digital signal processing, particularly using the so-called Sigma-Delta ($\Sigma\Delta$) approach, which quantizes the difference ($\Delta$) between the signal and its previously quantized value, and then integrates ($\Sigma$) the resulting output \cite{aziz1996overview,Baird,Galton}. 
This has the effect of shaping the quantization noise to higher frequencies, while the signal occupies the low end of the spectrum due to the oversampling.  Higher-order  $\Sigma\Delta$ modulators can be constructed that provide increased shaping of the quantization noise from low to high frequencies. Compared with a standard one-bit ADC, a $\Sigma\Delta$ ADC requires additional digital circuitry to implement the integration, but very little additional RF hardware. $\Sigma\Delta$ ADCs have been commonly used in process control and instrumentation applications, and more recently in the implementation of multi-channel beamformers for ultrasound imaging systems.

The concept of $\Sigma\Delta$ modulation can also be applied in the spatial as well as the temporal domain. In a spatial $\Sigma\Delta$ implementation, the difference signal is formed by subtracting the quantized output of one antenna’s RF chain from the signal at an adjacent antenna.  Coupled with spatial oversampling ({\em e.g.,} a uniform linear array with elements separated by less than one half wavelength), the quantization noise is shaped to higher {\em spatial} frequencies, and significantly reduced for signals arriving in a sector around broadside ($0^\circ$). Applying a phase shift to the feedback signal allows one to move the band of low quantization error to different angular regions. 

Relatively little research has focused on the spatial $\Sigma\Delta$ architecture. Prior related work has dealt with phased-array beamforming \cite{scholnik2004spatio,kriegerDense}, generalized structures for interference cancellation \cite{Venk2011}, and circuit implementations \cite{nikoofard,madan2017}. Applications of the idea to massive MIMO were first presented in \cite{Corey_Sig,baracspatial}, and more recently algorithms have been developed for channel estimation \cite{RaoSP19} and transmit precoding using $\Sigma\Delta$ DACs \cite{ShaoMLS19}. 

In this paper, we study the uplink spectral efficiency (SE) of a massive MIMO basestation (BS) that employs one-bit spatial $\Sigma\Delta$ quantization, and compare it with the performance achievable by systems with infinite resolution and standard one-bit quantization for maximum ratio combining (MRC) and zero-forcing (ZF) receivers. Past work on quantifying the SE for standard one-bit quantization ({\em e.g.}, \cite{Li,LiTSML17}) has relied on a vectorized version of the well-known Bussgang decomposition \cite{Bussgang}, which formulates an equivalent linear vector model for the array of non-linear quantizers assuming that the inputs to the quantizers are (at least approximately) jointly Gaussian. However, the vector Bussgang solution is not appropriate for the more complicated $\Sigma\Delta$ architecture, since it leads to a linear model that is inconsistent with the corresponding hardware implementation. Thus, we are led to derive an alternative linear model in which we apply a scalar version of the Bussgang approach to each quanitizer individually. This model is then used in turn to determine the overall sum SE. 

The results of the analysis indicate the significant gain of the $\Sigma\Delta$ approach compared with standard one-bit quantization for users that lie in the angular sector where the shaped quantization error spectrum is low. For MRC, the one-bit $\Sigma\Delta$ array performs essentially the same for such users as a BS with infinite resolution ADCs. The angular sectorization of users in the spatial domain is not necessarily a drawback in cellular implementations, where cells are typically split into $120^\circ$ regions using different arrays on the BS tower. In addition, there are many small-cell scenarios both indoors and outdoors where the targeted users are confined to relatively narrow angular sectors (auditoriums, plazas, arenas, etc.). Such situations will become even more prevalent as frequencies move to the millimeter wave band. However, the size of the sector of good performance for $\Sigma\Delta$ arrays depends on the amount of spatial oversampling. Unlike the temporal case, where oversampling factors of 10 or higher are not uncommon, the physical dimensions of the antenna and the loss due to increased mutual coupling for closely-spaced antennas places a limit on the amount of spatial oversampling that is possible in massive MIMO. Fortunately, our results indicate that spatial oversampling by factors of only 2-4 is sufficient to achieve good performance for angular sectors ranging from $80^\circ-150^\circ$.  Furthermore, the ability of the $\Sigma\Delta$ array to electronically steer the desired angular sector by means of the feedback phase shift provides desirable flexibility. For example, multiple sectors could be serviced in parallel with a single antenna array by deploying a bank of $\Sigma\Delta$ receivers tuned to different spatial frequencies, in order to cover a wider angular region.

In the next section we outline the basic system model, and provide some background on temporal $\Sigma\Delta$ modulation. In Section III, we introduce the spatial $\Sigma\Delta$ architecture. We develop an equivalent linear model and characterize this architecture in Section IV. The model is then applied to analyze the spectral efficiency of the $\Sigma\Delta$ array in Section V. While the analysis is conducted assuming that perfect channel state information (CSI) is available, we also discuss the impact of imperfect CSI in Section VI. Several simulation results are presented in Section VI, followed by our conclusions.

\emph{Notation}: 
We use boldface letters to denote vectors, and capitals to denote matrices. The symbols $(.)^*$, $(.)^T$, $(.)^H$, and $(.)^{\dagger}$ represent conjugate, transpose, conjugate transpose, and pseudo-inverse, respectively. A circularly-symmetric complex Gaussian (CSCG) random vector with zero mean and covariance matrix ${\boldsymbol{R}}$ is denoted $\boldsymbol{n}\sim\mathcal{CN}(\mathbf{0},\boldsymbol{{R}})$. The symbol $\|.\|$ represents the Euclidean norm. The identity matrix is denoted by $\mathit{\boldsymbol{I}}$, vector of
all ones by $\boldsymbol{1}$, and the expectation operator by $\mathbb{E}\left[.\right]$. We use $\mathrm{diag}\left(\boldsymbol{C}\right)$, $\mathrm{diag}\left(\boldsymbol{x}\right)$, and $\mathrm{diag}\left(x_1,\cdots,x_M\right)$ as the diagonal matrix formed from the diagonal entries of the square matrix $\boldsymbol{C}$, elements of vector $\boldsymbol{x}$, and scalars $x_1,\cdots,x_M$, respectively.  
For a complex value, $x=x_r+jx_i$, we define $x_r=\mathfrak{Re}\left[x\right]$ and $x_i=\mathfrak{Im}\left[x\right]$.


\section{System Model}\label{sec:SYSTEM MODEL}

Consider the uplink of a single-cell multi-user MIMO system consisting of $K$ single-antenna users that send their signals simultaneously to a BS equipped with a uniform linear array (ULA) with ${M}$ antennas. The $M\times 1$ signal received at the BS from the $K$ users is given by
\begin{equation}\label{system-model}
  {\boldsymbol{x}}=\boldsymbol{G}\boldsymbol{P}^{\frac{1}{2}}
  \boldsymbol{{s}} + {\boldsymbol{n}},
\end{equation}
where $\boldsymbol{G}=\left[\boldsymbol{g}_1,\cdots,\boldsymbol{g}_K\right]\in\mathbb{C}^{M\times K}$ is the channel matrix between the users and the BS and $\boldsymbol{P}$ is a diagonal matrix whose $k$th diagonal element, $p_k$, represents the transmitted power of the $k$th user. The symbol vector transmitted by the users is denoted by $\boldsymbol{s}\in\mathbb{C}^{K\times 1}$ where $\mathbb{E}\left\{\boldsymbol{s}\boldsymbol{s}^H\right\}=\boldsymbol{I}_K$ and is drawn from a circularly symmetric complex Gaussian (CSCG) codebook independent of the other users,
and, ${\boldsymbol{n}}\sim\mathcal{CN}\left(\boldsymbol{0},\sigma_n^2\boldsymbol{I}_{M}\right)$ denotes additive CSCG receiver noise at the BS. 

We consider a physical channel model described in the angular domain and comprised of $L$ paths for each user with azimuth angular spread $\Theta$ \cite{Ngo_PHY}. In particular, for the $k$th user, the channel vector is modeled as
\begin{equation}\label{channel model_phy}
\boldsymbol{g}_k=\sqrt{\frac{\beta_k}{L}}\boldsymbol{A}_k\boldsymbol{h}_k,	
\end{equation}
where $\boldsymbol{A}_k$ is an $M\times L$ matrix whose $\ell$th column is the array steering vector corresponding to the direction of arrival (DoA) $\theta_{k\ell}\in\theta_0+\left[-\frac{\Theta}{2},\frac{\Theta}{2}\right]$, $\beta_k$ models geometric attenuation and shadow fading from the $k$th user to the BS, and the elements of $\boldsymbol{h}_k\in\mathbb{C}^{L\times 1}$ are assumed to be distributed identically and independently as $\mathcal{CN}\left({0},1\right)$, and model the fast fading propagation. For a ULA, the steering vector for a signal with DoA $\theta_{k\ell}$ is expressed as
\begin{equation}\label{steering vector}
\boldsymbol{a}\left(u_{k\ell}\right)=\begin{bmatrix}
        1,z_{k\ell}^{-1},\cdots,z_{k\ell}^{-(M-1)}
\end{bmatrix}^T,	
\end{equation}
where $u_{k\ell}=\mathrm{sin}\left(\theta_{kl}\right)$, $z_{k\ell}=e^{j\omega_{s_{k\ell}}}$, and thus $\omega_{s_{k\ell}}=2\pi\frac{d}{\lambda}u_{k\ell}$ represents the spatial frequency assuming antenna spacing $d$ and wavelength $\lambda$. 

In a standard implementation involving one-bit quantization, each antenna element at the BS is connected to a one-bit ADC. In such systems, the received baseband signal at the $m$th antenna becomes
\begin{equation}\label{quantizer}
    y_m = \mathcal{Q}_m\left(x_m\right) \; ,
\end{equation}
where $\mathcal{Q}_m\left(.\right)$ denotes the one-bit quantization operation which is applied separately to the real and imaginary parts as
\begin{equation}
\mathcal{Q}_m\left(x_m\right)=\alpha_{m,r}\mathrm{sign}\left(\mathfrak{Re}\left(x_m\right)\right)+j\alpha_{m,i}\mathrm{sign}\left(\mathfrak{Im}\left(x_m\right)\right) \; ,
\end{equation}
where $\alpha_{m,r}$ and $\alpha_{m,i}$ represent the output voltage levels of the one-bit quantizer. We will allow these levels to be a function of the antenna index $m$, unlike most prior work which assumes that the output levels are the same for all antennas. The necessity for this more general approach will become apparent later\footnote{While the one-bit ADC output levels will be optimized, this is a one-time optimization and the values do not change as a function of the user scenario or channel realization. Thus the ADCs are still truly ``one-bit.''}. Finally, the received baseband signal at the BS is given by
\begin{equation}\label{vector-quantizer}
    \boldsymbol{y}=\mathcal{Q}\left(\boldsymbol{x}\right) = 
    \begin{bmatrix}
           \mathcal{Q}_1\left(x_1\right),
           \mathcal{Q}_2\left(x_2\right),
           \cdots,
           \mathcal{Q}_M\left(x_M\right)
    \end{bmatrix}^T.
\end{equation}

\section{$\Sigma\Delta$ Architecture}\label{sec:sd-architecture}
\subsection{Temporal $\Sigma\Delta$ Modulation}\label{sec: temporal-sigma-delta}
In this subsection, we elaborate on temporal $\Sigma\Delta$ modulation to clarify the noise shaping characteristics of this technique. 
\begin{figure}
    \centering
    \subfigure[\mbox{\rm }]
    {
        \includegraphics[width=0.5\textwidth]{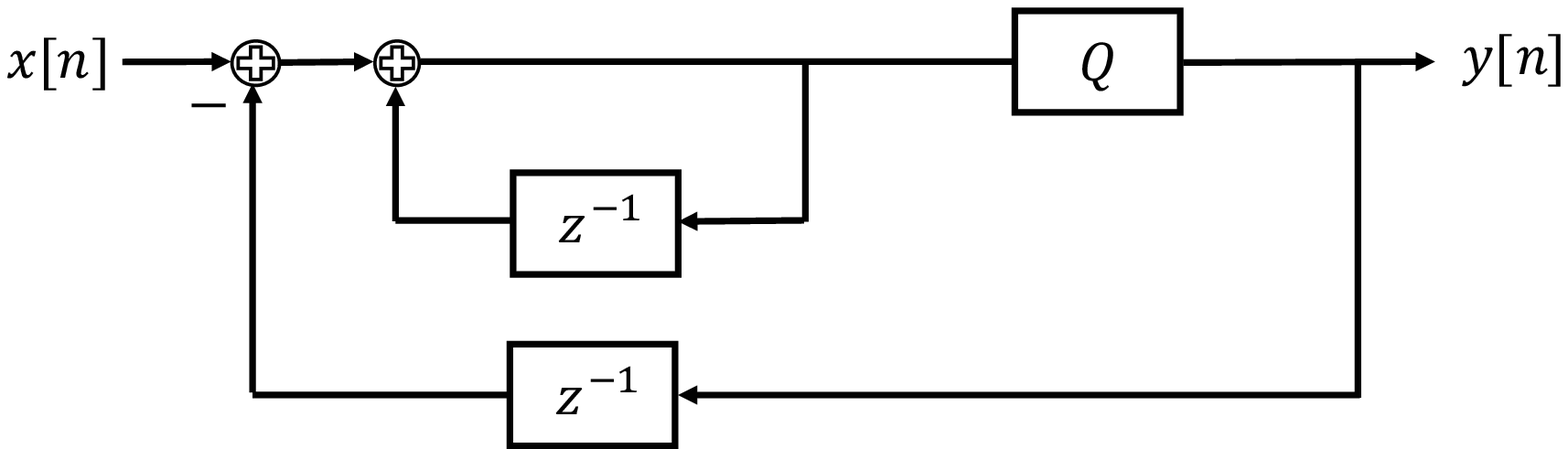}
        \label{pic1}
    }
    \\
    \subfigure[\mbox{\rm }]
    {
        \includegraphics[width=0.5\textwidth]{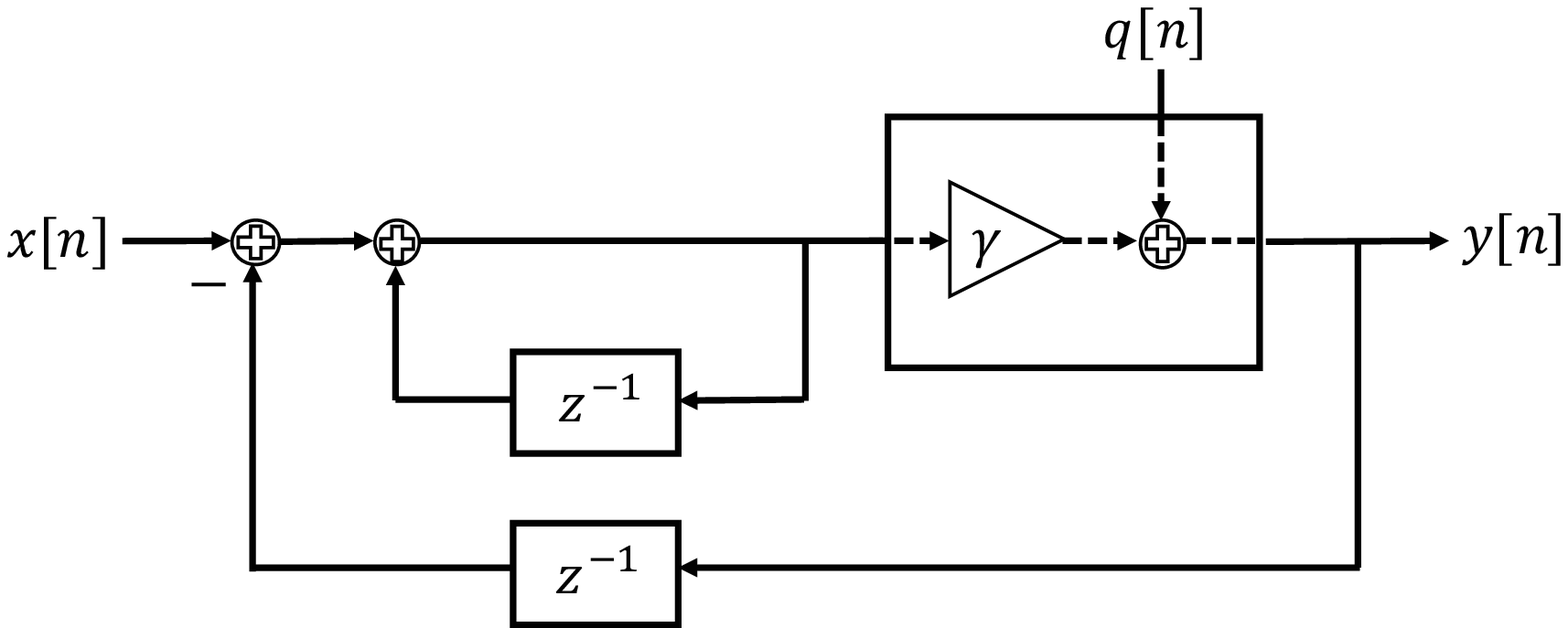}
        \label{pic2}
    }
    \caption{(a) Block diagram for temporal $\Sigma\Delta$ modulator. (b) With equivalent linear model for quantization.}
    \label{pic12}
\end{figure}
Fig. \ref{pic1} shows a block diagram representing the temporal $\Sigma\Delta$ modulator. To shape the quantization noise, the output signal is fed back and subtracted form the input ($\Delta$-stage), and then this error is integrated ($\Sigma$-stage). To characterize the transfer function of this non-linear system, we substitute the one-bit quantizer with the equivalent linear model depicted in Fig. \ref{pic2}. In Fig. \ref{pic2}, the \emph{equivalent gain of the non-linear device}, $\gamma$, is a function of the quantizer's output level and is chosen to make the input of the quantizer uncorrelated with the equivalent quantization noise, $q[n]$. This is a common approach for modeling non-linear systems \cite{Booton,Bussgang}. Therefore,
the input-output relationship of the $\Sigma\Delta$ quantizer can then be written as
\begin{equation}\label{transfer-function}
    Y\left(z\right)=\frac{\gamma}{1-\left(1-\gamma\right)z^{-1}}X\left(z\right) + \frac{\left(1-z^{-1}\right)}{1-\left(1-\gamma\right)z^{-1}}Q\left(z\right),
\end{equation}
where $X\left(z\right)=\sum_{n=0}^{\infty}{x\left[n\right]z^{-n}}$ denotes the $z$-transform. Simply stated, the objective of $\Sigma\Delta$ modulation is to pass the signal through an all-pass filter and the quantization noise through a high-pass filter. This objective can be realized by selecting the output voltage level of the quantizer such that $\gamma\approx1$. Since commercial quantizers are provided with a built-in automatic gain control (AGC), the $\gamma\approx 1$ condition is inherently satisfied in implementations of temporal $\Sigma\Delta$ modulators, and hence this issue is not generally discussed in the literature. However, as we show in the next subsection, the choice of the scaling factor is critical in the mathematical modeling of spatial $\Sigma\Delta$ architectures, and we derive a criterion for addressing this issue.

\subsection{One-Bit Spatial $\Sigma\Delta$ Modulation}\label{sec: spatial-sigma-delta}
As mentioned earlier, the basic premise of temporal $\Sigma\Delta$ modulation can be adopted in the angle domain, in order to {\em spatially} shape the quantization noise in a desired way. Instead of forming the $\Delta$ component using a delayed sample of the quantized input as in the temporal case, we use the quantization error signal from an adjacent antenna. A direct transfer of the temporal $\Sigma\Delta$ idea to the angle domain as in \cite{Corey_Sig,baracspatial} pushes the quantization noise to higher spatial frequencies, which correspond to DoAs away from the array broadside ($|\theta|\gg 0^\circ$), while the oversampling (reduced $d/\lambda$) pushes signals of interest near broadside closer to zero spatial frequency. However, by phase-shifting the quantization error in the feedback loop prior to the $\Delta$ stage, a $\Sigma\Delta$ frequency response can be obtained in which the quantization error is shaped away from a band of frequencies not centered at zero. This bandpass approach has been proposed for both the temporal ({\em e.g.,} see \cite{aziz1996overview}) and spatial \cite{ShaoMLS19} versions of the $\Sigma\Delta$ architecture.

\begin{figure}
\centering
\includegraphics[width=0.5\textwidth]
{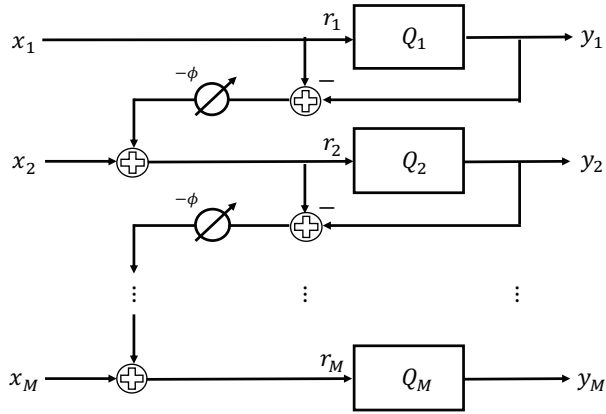}
\caption{Spatial $\Sigma\Delta$ architecture. 
}
\label{pic3}
\end{figure}
Fig. \ref{pic3} shows the architecture of an angle-steered $\Sigma\Delta$ array. 
Using Fig.~\ref{pic3} and equation~(\ref{s-d-relation}) at the top of the next page, we can formulate a compact input-output description of the spatial $\Sigma\Delta$ array by defining 
\begin{figure*}[!t]
\begin{equation}\label{s-d-relation}
y_m=
\begin{dcases*}
\begin{aligned}
& \mathcal{Q}_1\left(x_1\right)
& m=1 \\
& \mathcal{Q}_m\left(x_{m}+ e^{-j\phi}\left(x_{m-1} + e^{-j\phi}\left(\cdots\left(x_2+e^{-j\phi}\left(x_1-y_1\right)-y_2\right)\cdots\right)-y_{m-1}\right)\right) &   m>1\\
\end{aligned}
\end{dcases*}
\end{equation}
\hrulefill
    \vspace*{4pt}
    \end{figure*}
\begin{equation}\label{U-matrix}
\boldsymbol{U}=
  \left[ {\begin{array}{ccccc}
   1 & & & & \\
   e^{-j\phi} & 1 & & & \\
   \vdots & \ddots & \ddots & & \\
   e^{-j\left(M-1\right)\phi} & \cdots & e^{-j\phi} & 1 & \\
  \end{array} } \right]
\end{equation}
\begin{equation}\label{V-matrix}
    \boldsymbol{V}=\boldsymbol{U}-\boldsymbol{I}_M,
\end{equation}
and expressing the input to the quantizers as
\begin{equation}\label{r_def}
    \boldsymbol{r}=\boldsymbol{U}\boldsymbol{x}-\boldsymbol{V}\boldsymbol{y}.
\end{equation}
The output of the angle-steered one-bit $\Sigma\Delta$ array is then defined by
\begin{equation}\label{s-d-vector}
    \boldsymbol{y}=\mathcal{Q}\left(\boldsymbol{r}\right).
\end{equation}

\section{Characterizing the Spatial $\Sigma\Delta$ Architecture}
\subsection{Linear Model}\label{sec:linear-model}
To analyze the performance of spatial $\Sigma\Delta$ processing, analogous to temporal $\Sigma\Delta$, we will represent the one-bit quantization operation in~(\ref{s-d-vector}) with an equivalent linear model as follows:
\begin{equation} \label{bussgang-IO}
    \boldsymbol{y}=\mathcal{\mathcal{Q}}\left(\boldsymbol{r}\right)=\boldsymbol{\mathit{\Gamma}}\boldsymbol{r}+\boldsymbol{q},
\end{equation}
where $\boldsymbol{\mathit{\Gamma}}$ is an $M\times M$ matrix and $\boldsymbol{q}$ denotes the effective quantization noise. The value of $\boldsymbol{\mathit{\Gamma}}$ that makes the equivalent quantization noise, $\boldsymbol{q}$, uncorrelated with $\boldsymbol{r}$ is $\boldsymbol{\mathit{\Gamma}}_0=\boldsymbol{R}_{\boldsymbol{r}\boldsymbol{y}}^{H}\boldsymbol{R}_{\boldsymbol{r}}^{-1}$. For the case where the elements of $\boldsymbol{r}$ are all jointly Gaussian, the computation of $\boldsymbol{R}_{\boldsymbol{r}\boldsymbol{y}}$ is possible by resorting to the Bussgang theorem\footnote{The result can also be extended to cases where $\boldsymbol{r}$ belongs to a limited class of distributions, see \cite{Barret} for details} \cite{Bussgang}. This was the approach used in \cite{Li,LiTSML17} for a massive MIMO implementation with standard one-bit quantization, and the resulting $\boldsymbol{\mathit{\Gamma}}_0$ was a diagonal matrix. 

For the case of the $\Sigma\Delta$ architecture, even if the matrix $\boldsymbol{\mathit{\Gamma}}_0$ could be computed, this decomposition would not be of interest, for at least two reasons. First, the equivalent quantization noise $\boldsymbol{q}$ that results from setting $\boldsymbol{\mathit{\Gamma}}=\boldsymbol{\mathit{\Gamma}}_0$ in~\eqref{bussgang-IO} bears no connection to the quantization error fed from one antenna to the next as shown in Fig.~\ref{pic3}. Setting $\boldsymbol{\mathit{\Gamma}}=\boldsymbol{\mathit{\Gamma}}_0$ would produce a model in which $r_m$ and $q_{m-1}$ are uncorrelated, but it is clear from Fig.~\ref{pic3} that $r_m$ for the $\Sigma\Delta$ array directly depends on the quantization error from the $(m-1)$-th stage.
Second, $\boldsymbol{\mathit{\Gamma}}_0$ cannot be a diagonal matrix\footnote{
If $\boldsymbol{\mathit{\Gamma}}_0$ were diagonal, it could be made equal to the identity matrix by a proper scaling of each $y_m$. However, $\boldsymbol{\mathit{\Gamma}}_0$ can never be the identity matrix because this implies that $r_m=x_m-e^{-j\phi}q_{m-1}$, while simultaneously $r_m$ is uncorrelated with $q_{m-1}$, which is impossible.}, unlike the standard one-bit quantization case considered in \cite{Li}. The presence of off-diagonal elements in $\boldsymbol{\mathit{\Gamma}}_0$ implies that the model in (\ref{bussgang-IO}) represents the output of each quantizer as a linear combination of the inputs to that quantizer as well as other quantizers in the array. Such a model does not have an apparent connection with the scheme in Fig.~\ref{pic3}, where each quantizer 
produces its output depending only on its input alone. These inconsistencies between the mathematical model based on $\boldsymbol{\mathit{\Gamma}}=\boldsymbol{\mathit{\Gamma}}_0$ and the physical block diagram of the $\Sigma\Delta$ array in Fig.~\ref{pic3} are the result of attempting to force $\boldsymbol{r}$ and $\boldsymbol{q}$ to be uncorrelated, when the architecture is actually propagating the quantization error from one stage to the next.
\color{black}

Consequently, in order to derive an appropriate model for the analysis of the $\Sigma\Delta$ architecture, we propose to apply the Bussgang decomposition to each quantizer individually. In particular, we formulate the model in~(\ref{bussgang-IO}) using a matrix $\boldsymbol{\mathit{\Gamma}}=\mathrm{diag}\left(\gamma_1,\dots,\gamma_M\right)$ that is forced to be diagonal. This is equivalent to imposing a model in which $\boldsymbol{r}$ and $\boldsymbol{q}$ are uncorrelated component-wise: $\mathbb{E}\left[r_m q_m^*\right]=0$, which is the same criterion used to generate the model for the scalar case in Section~\ref{sec: temporal-sigma-delta}. The elements of $\boldsymbol{\mathit{\Gamma}}$ are given by
\begin{equation}\label{gamma_m2}
    \mathit{\gamma}_{m}=
    \frac{\mathbb{E}\left[r_my_m^*\right]}{\mathbb{E}\left[\left|r_m\right|^2\right]}=\alpha_m\frac{\mathbb{E}\left[\left|\mathfrak{Re}\left[r_m\right]\right|+\left|\mathfrak{Im}\left[r_m\right]\right|\right]}{\mathbb{E}\left[\left|r_m\right|^2\right]},
\end{equation}
where in the last equality and from now on, we assume that $r_m$ is circularly symmetric. This assumption implies that the quantizer output levels are identical for the real and imaginary parts, and thus we use $\alpha_m$ to represent both $\alpha_{m,r}$ and $\alpha_{m,i}$. 

As we will see later on, since the elements of $\boldsymbol{\mathit{\Gamma}}$ depend only on the signals at one stage of the $\Sigma\Delta$ architecture, they are much easier to compute than the elements of $\boldsymbol{\mathit{\Gamma}}_0$. Moreover, the resulting decomposition is consistent with Fig.~\ref{pic3}. Given that no precondition is imposed on the correlation $\mathbb{E}\left[r_m q_l^*\right]$ for $m \neq l$, the model is compatible with the fact that the quantization noise of one stage appears in subsequent stages.
\color{black}

Plugging (\ref{bussgang-IO}) into (\ref{r_def})  and using some algebraic manipulations, we {obtain the following mathematical model for the $\Sigma\Delta$ architecture:}
\begin{equation}\label{s-d-general-form}
    \boldsymbol{y}=\big(\boldsymbol{I}+\boldsymbol{\mathit{\Gamma}}\boldsymbol{V}\big)^{-1}\boldsymbol{\mathit{\Gamma}}\boldsymbol{U}\boldsymbol{x}+\big(\boldsymbol{I}+\boldsymbol{\mathit{\Gamma}}\boldsymbol{V}\big)^{-1}\boldsymbol{q}.
\end{equation}
Equation~(\ref{s-d-general-form}) is the spatial $\Sigma\Delta$ equivalent to the temporal domain $\Sigma\Delta$ description in~(\ref{transfer-function}). Similar to the temporal case, (\ref{s-d-general-form}) indicates that $\boldsymbol{\mathit{\Gamma}}=\boldsymbol{I}$ should hold for the spatial $\Sigma\Delta$ array to work as desired, that is, to pass $\boldsymbol{x}$ and $\boldsymbol{q}$ through spatial all-pass and high-pass filters, respectively. If $\boldsymbol{\mathit{\Gamma}}=\boldsymbol{I}$, then (\ref{s-d-general-form}) becomes
\begin{equation}\label{shaped-equation}
    \boldsymbol{y} = \boldsymbol{x} + \boldsymbol{U}^{-1}\boldsymbol{q} \; ,
\end{equation}
and the $m$-th element of $\boldsymbol{y}$ is expressed as 
\begin{equation}
    y_m=x_m+\left(q_m-e^{-j\phi}q_{m-1}\right) \; ,
\end{equation}
which explicitly shows the quantization noise-shaping characteristic of the spatial $\Sigma\Delta$ architecture.
\color{black}
The only task remaining to complete our proposed linear model is to calculate the power of the equivalent quantization noise.
The condition $\boldsymbol{\mathit{\Gamma}}=\boldsymbol{I}$ for the adequate operation of the $\Sigma\Delta$ scheme determines the quantization levels that have to be set. Setting (\ref{gamma_m2}) equal to $1$, we obtain the optimum value of $\alpha_m$:
\begin{equation}\label{alpha_opt}
    \alpha_{m}^{\star}=\frac{\mathbb{E}\left[\left|r_m\right|^2\right]}{\mathbb{E}\left[\left|\mathfrak{Re}\left[r_m\right]\right|+\left|\mathfrak{Im}\left[r_m\right]\right|\right]}=\frac{\mathbb{E}\left[\left|\mathfrak{Re}\left[r_m\right]\right|^2\right]}{\mathbb{E}\left[\left|\mathfrak{Re}\left[r_m\right]\right|\right]}.
\end{equation}

It is worth noting that (\ref{alpha_opt}) is different from 
\begin{equation}\label{alpha_llm}
    \alpha_{m}=\mathbb{E}\left[\left|\mathfrak{Re}\left[r_m\right]\right|\right],
\end{equation}
which leads to the Lloyd-Max one-bit quantizer that minimizes the mean-squared-error (MSE) between the input and the output of the quantizer. However, the Lloyd-Max approach makes the quantization error uncorrelated with the quantizer output, but not with the input.

While the expression derived in~(\ref{alpha_opt}) is useful, it is difficult to analytically evaluate the expectations in closed form, and it is not clear how the output level could be tuned using analog processing in the RF chain (e.g., via an AGC or some other type of calibration). To address this issue, we use the assumption that $r_m$ is Gaussian inherent in the Bussgang decomposition to find an approximation for $\alpha_m^\star$ that is easier to deal with, both for the subsequent mathematical analysis and from the viewpoint of a hardware implementation. The validity of the approximation will be apparent in the numerical examples presented later. If $r_m$ is Gaussian, we can write
\begin{equation}\label{alphaapprox}
    \alpha_{m}^{\star} = \frac{\sqrt{\pi \mathbb{E}\left[\left|r_m\right|^2\right]}}{2}.
\end{equation}
In the discussion below, we show how to express~(\ref{alphaapprox}) in terms of the statistics of the array output $\boldsymbol{x}$, which provides an analytical solution and clarifies how the quantizer output levels could be set in a practical setting.
\color{black}

\subsection{Quantization Noise Power}\label{quantization-noise-power}
In this section, we calculate the power of the effective quantization noise and the power of the quantizers' inputs, which is needed to properly set the output levels using (\ref{alphaapprox}). With $\boldsymbol{\mathit{\Gamma}}=\boldsymbol{I}$, (\ref{bussgang-IO}) becomes
\begin{equation}\label{io-eq}
    \boldsymbol{y}=\boldsymbol{r}+\boldsymbol{q}.
\end{equation}
Since $r_m$ and $q_m$ are uncorrelated, and using (\ref{alphaapprox}), we obtain 
\begin{equation}\label{q-var}
 \mathbb{E}\left[|q_m|^2\right]=
     \mathbb{E}\left[|y_m|^2\right] - \mathbb{E}\left[|r_m|^2\right]=\left(\frac{\pi}{2}-1\right)\mathbb{E}\left[|r_m|^2\right].
\end{equation} 
 To determine $\mathbb{E}\left[|r_m|^2\right]$, we substitute (\ref{io-eq}) 
 into (\ref{r_def}), so that 
\begin{equation}\label{q-recursion}
    \boldsymbol{r}=\boldsymbol{x} - \boldsymbol{U}^{-1}\boldsymbol{V}\boldsymbol{q}.
\end{equation}
It can be shown that
\begin{equation}\label{Zi1}
    \boldsymbol{U}^{-1}\boldsymbol{V}=e^{-j\phi}\boldsymbol{Z}_{-1} \; ,
\end{equation}
where\footnote{Note that $\boldsymbol{Z}_{-1}$ is the spatial domain equivalent of the delay operator $z^{-1}$ for the z-transform in the time domain.}
\begin{equation}\label{Zi2}
    \boldsymbol{Z}_{-1}=\left[ {\begin{array}{ccccc}
   0 & & & & \\
   1 & 0 & & & \\
   \vdots & \ddots & \ddots & & \\
   0 & \ddots & 1 & 0 & \\
  \end{array} } \right].
\end{equation}
Moreover, following the same reasoning as in Appendix A of \cite{Li}, it can be shown that $\mathbb{E}\left[x_{m'}q_m^*\right]\approx 0,~\forall m,m'\in\mathcal{M}=\left\{1,\cdots,M\right\}$. This results in $\boldsymbol{R}_{\boldsymbol{qx}}\approx\boldsymbol{0}$. 
 \color{black}
 Therefore,
\begin{equation}\label{recursion1}
    \boldsymbol{R}_{\boldsymbol{r}}=\boldsymbol{R}_{\boldsymbol{x}} + \boldsymbol{Z}_{-1}\boldsymbol{R}_{\boldsymbol{q}}\boldsymbol{Z}^{H}_{-1} \;.
\end{equation}

Eq.~(\ref{recursion1}) implies that
\begin{equation}\label{recursion2}
\mathbb{E}\left[|r_m|^2\right] =
\begin{dcases*}
\begin{aligned}
& \mathbb{E}\left[|x_m|^2\right]
& m=1 \\
& \mathbb{E}\left[|x_m|^2\right] + \mathbb{E}\left[|q_{m-1}|^2\right] &   m>1\\
\end{aligned}
\end{dcases*}
\end{equation}
Substituting (\ref{q-var}) into (\ref{recursion2}) and noting that $\mathbb{E}\left[|r_1|^2\right]=\mathbb{E}\left[|x_1|^2\right]$, we obtain the following recursive equality to calculate $\mathbb{E}\left[|r_m|^2\right]$ for $m>1$:
\begin{equation}\label{recursive-eq}
    \mathbb{E}\left[|r_m|^2\right]=\mathbb{E}\left[|x_m|^2\right]+\left(\frac{\pi}{2} -1\right)\mathbb{E}\left[|r_{m-1}|^2\right].
\end{equation}
Let 
\begin{equation}
\boldsymbol{p}_{\boldsymbol{\chi}}=
\begin{bmatrix}
        \mathbb{E}\left[|\chi_1|^2\right], \mathbb{E}\left[|\chi_2|^2\right],\cdots,\mathbb{E}\left[|\chi_M|^2\right]
\end{bmatrix}^T,
\end{equation}
where $\boldsymbol{\chi}$ can be any element of the set $\boldsymbol{\chi}\in\left\{\boldsymbol{r},\boldsymbol{x},\boldsymbol{q}\right\}$.
Then, using (\ref{q-var}) and (\ref{recursive-eq}), we have
\begin{equation}\label{rxpower}
    \boldsymbol{p}_{\boldsymbol{r}}=\boldsymbol{\Pi}\boldsymbol{p}_{\boldsymbol{x}}
\end{equation}
\begin{equation}\label{qnoise-power}
    \boldsymbol{p}_{\boldsymbol{q}}=\left(\frac{\pi}{2}-1\right)\boldsymbol{\Pi}\boldsymbol{p}_{\boldsymbol{x}} \; ,
\end{equation}
where 
\begin{equation}
\boldsymbol{\Pi}=    
    \left[
    \begin{array}{cccccc}
    \begin{matrix} 
    1 & & & & & \boldsymbol{0} \cr
    \left(\frac{\pi}{2}-1\right) & 1 & & & &  \cr
    \vdots & \ddots & 1 & & & \cr \left(\frac{\pi}{2}-1\right)^{m} & \ddots& \ddots & \ddots & & \cr \vdots& \ddots & \ddots& \ddots & \ddots & \cr \left(\frac{\pi}{2}-1\right)^{M-1} &\cdots & \left(\frac{\pi}{2}-1\right)^{m}& \cdots&\left(\frac{\pi}{2}-1\right) & 1
    \end{matrix}
    \end{array} \right]. 
\end{equation}

Equation~(\ref{rxpower}) shows that the calculation of $\mathbb{E}\left[\left|r_m\right|^2\right]$ needed in~(\ref{alphaapprox}) can be formulated in terms of the power of the antenna outputs $\mathbb{E}\left[\left|x_m\right|^2\right]$, for which simple expressions exist from~(\ref{system-model}). This further implies that control of $\mathbb{E}\left[\left|x_m\right|^2\right]$ via an AGC would allow the quantizer output levels to be set without feedback from the digital baseband. In the following remark, we show that, using the optimal quantizer output settings, the power of the quantization noise does not grow with $m$ despite the fact that it is propagated from one antenna to the next. 

\begin{rem}
Eq. (\ref{qnoise-power}) implies that, by appropriately selecting the quantizers' output levels, the quantization noise power does not increase without bound. In particular, consider the case where the power of the received signal is constant over the array elements, i.e., $\boldsymbol{p}_{\boldsymbol{x}}=p_x\boldsymbol{1}$. Then,
\begin{equation}
    \mathbb{E}\left[|q_m|^2\right]=\left(\frac{\pi}{2}-1\right)\frac{1-\left(\frac{\pi}{2}-1\right)^m}{1-\left(\frac{\pi}{2}-1\right)}p_x\;{\xrightarrow[m\rightarrow\infty]{}}\;\frac{\frac{\pi}{2}-1}{2-\frac{\pi}{2}}p_x \; ,
\end{equation}
which shows that, in the limit of a large number of antenna elements, the quantization noise power converges to a constant value of approximately 1.33 times the input power.
\end{rem}

\subsection{Quantization Noise Power Density}

In the time domain, it is well-known that sampling a band-limited signal by a rate $N$ times larger than the Nyquist rate and down-sampling after quantization can reduce the in-band quantization noise power by a factor of $1/N$ and $1/N^3$ for standard one-bit and $\Sigma\Delta$ modulation, respectively \cite{Oppenheim}. 
In this subsection, we look for a similar behaviour for quantization across an array in space. More precisely, we want to quantify how spatial oversampling, i.e., decreasing the antenna spacing, $d/\lambda$, (or equivalently, increasing the number of antennas for space-constrained arrays) can reduce the quantization noise power for the in-band angular spectrum. To do so, we define the quantization noise power density as 
\begin{equation}\label{noise-density}
   \rho_q\left(u\right)\triangleq\frac{1}{ M}{\boldsymbol{a}\left(u\right)^H\boldsymbol{R}~\boldsymbol{a}\left(u\right)}, 
\end{equation}
where $\boldsymbol{R}$ is the covariance matrix of the quantization noise. To differentiate the two cases, we denote the covariance matrix of the quantization noise for standard one-bit quantization as $\boldsymbol{R}_{\boldsymbol{q}_1}$, and the covariance of the $\Sigma\Delta$ quantization noise as $\boldsymbol{R}_{\boldsymbol{q}_{\Sigma\Delta}}$. Expressions for these covariance matrices will be derived later in this subsection. Hence, the normalized received quantization noise power over some angular region, $\Theta$, is given by\footnote{To simplify the calculation of the quantization noise power, we assume without loss of generality that the $\Sigma\Delta$ array is steered to broadside ($\theta=0$).}
\begin{equation}\label{noise-power}
     \mathcal{P}_q=\frac{1}{2\delta}\int_{-\delta}^{\delta}{\rho_q\left(u\right)du},
\end{equation}
where $\delta=\mathrm{sin}\left(\frac{\Theta}{2}\right)$. Next we find $\mathcal{P}_q$ for standard one-bit and $\Sigma\Delta$ quantization.

\subsubsection{One-bit Quantization}
Unlike \cite{Li}, for standard one-bit quantization, we choose the quantizer output levels as $\alpha_m={\sqrt{\pi\mathbb{E}\left[\left|x_m\right|^2\right]}}/{2}$ so that $y_m=\mathcal{Q}\left(x_m\right)=x_m+q_m$. This causes no loss of generality for standard one-bit quantization, since the value of the quantizer output has no impact on the performance of the resulting system. Therefore, the covariance matrix of the quantization noise can be written as
\begin{equation}
    \boldsymbol{R}_{\boldsymbol{q}_1}=\boldsymbol{R}_{\boldsymbol{y}} - \boldsymbol{R}_{\boldsymbol{x}} \; , 
\end{equation}
where the arc-sine law \cite{Vleck66,Jacovitti94} is used to obtain 
\begin{equation}\label{y1-cov}
 \boldsymbol{R}_{\boldsymbol{y}}=  \mathrm{diag}\left(\boldsymbol{R}_{\boldsymbol{x}}\right)^{\frac{1}{2}}\mathrm{sin}^{-1}\left(\boldsymbol{\Upsilon}\right)~\mathrm{diag}\left(\boldsymbol{R}_{\boldsymbol{x}}\right)^{\frac{1}{2}}, 
\end{equation}
and
\begin{equation}
    \boldsymbol{\Upsilon}=\mathrm{diag}\left(\boldsymbol{R}_{\boldsymbol{x}}\right)^{-\frac{1}{2}}\mathfrak{Re}\left(\boldsymbol{R}_{\boldsymbol{x}}\right)\mathrm{diag}\left(\boldsymbol{R}_{\boldsymbol{x}}\right)^{-\frac{1}{2}} +\\ j\mathrm{diag}\left(\boldsymbol{R}_{\boldsymbol{x}}\right)^{-\frac{1}{2}}\mathfrak{Im}\left(\boldsymbol{R}_{\boldsymbol{x}}\right)\mathrm{diag}\left(\boldsymbol{R}_{\boldsymbol{x}}\right)^{-\frac{1}{2}}.
\end{equation}
Note that the arc-sine in~(\ref{y1-cov}) is applied separately to each element of the matrix argument, and also separately to the real and imaginary parts of the matrix elements.

From \cite{Li}, we have that $\mathrm{diag}\left(\boldsymbol{R}_{\boldsymbol{y}}\right)=\frac{\pi}{2}\mathrm{diag}\left(\boldsymbol{R}_{\boldsymbol{x}}\right)$. Since the off-diagonal elements of $\boldsymbol{\Upsilon}$ are small, we use the approximation $\mathrm{sin}^{-1}\left(x\right)\approx \zeta x$, where $\zeta>1$, to obtain 
\begin{equation}\label{q1-cov-approx}
    \boldsymbol{R}_{\boldsymbol{q}_1}\approx\left(\zeta-1\right)\boldsymbol{R}_{\boldsymbol{x}} +\left(\frac{\pi}{2}-\zeta\right)\mathrm{diag}\left( \boldsymbol{R}_{\boldsymbol{x}}\right). 
\end{equation}
Moreover, from (\ref{system-model}), $\boldsymbol{R}_{\boldsymbol{x}}$ becomes
\begin{equation}
  \boldsymbol{R}_{\boldsymbol{x}}= \sum_{k=1}^{K}{p_k\beta_k\frac{1}{L}\sum_{\ell=1}^{L}{\boldsymbol{a}\left(u_{k\ell}\right)\boldsymbol{a}\left(u_{k\ell}\right)^H}+\sigma_n^2\boldsymbol{I}}, 
\end{equation}
where for $L\gg 1$, $u_{k\ell}$ can be taken as a random variable uniformly distributed in $\left[-\delta,\delta\right]$. That is,
\begin{equation} \label{sum_exp_approx}
    \frac{1}{L}\sum_{\ell=1}^{L}{\boldsymbol{a}\left(u_{k\ell}\right)\boldsymbol{a}\left(u_{k\ell}\right)^H}\approx\mathbb{E}\left[\boldsymbol{a}\left(u\right)\boldsymbol{a}\left(u\right)^H\right]=\\
    \frac{1}{2\delta}\int_{-\delta}^{\delta}{\boldsymbol{a}\left(u\right)\boldsymbol{a}\left(u\right)^H du}.
\end{equation}
Therefore,
\begin{equation}\label{Rx-int}
    \boldsymbol{R}_{\boldsymbol{x}}=\sum_{k=1}^{K}{p_k\beta_k\frac{1}{2\delta}\int_{-\delta}^{\delta}{\boldsymbol{a}\left(u\right)\boldsymbol{a}\left(u\right)^H du}}+\sigma_n^2\boldsymbol{I}.
\end{equation}

Now we are ready to calculate the standard one-bit quantization noise power, $\mathcal{P}_{q_1}$.
\begin{pro}\label{rho-q1-pro}
The normalized quantization noise power for standard one-bit quantization is 
\begin{equation}\label{rho_q1}
    \mathcal{P}_{q_1}=\left(\zeta-1\right)\times\\
    \left[\sigma_n^2+\frac{1}{M}\sum_{k=1}^{K}{p_k\beta_k\sum_{n=0}^{M-1}{\sum_{m=0}^{M-1}{\mathrm{sinc}^2\left(2\pi\frac{d}{\lambda}\left({m}-{n}\right)\delta\right)}}}\right]+\\
    \frac{\frac{\pi}{2}-\zeta}{M}\mathrm{Tr}\left[\boldsymbol{R}_{\boldsymbol{x}}\right],
\end{equation}
where $\mathrm{sinc}\left(x\right)\triangleq\frac{\mathrm{sin}\left(x\right)}{x}$.
\end{pro}
\begin{proof}
Plugging (\ref{Rx-int}) into (\ref{q1-cov-approx}) results in  
\begin{equation}
    \mathcal{P}_{q_1}=\left(\zeta-1\right)\times\\
    \left[\sigma_n^2 + \frac{1}{4\delta^2 M}\sum_{k=1}^{K}{p_k\beta_k\iint\limits_{-\delta}^{\delta}{\left|\boldsymbol{a}\left(v\right)^H\boldsymbol{a}\left(u\right)\right|^2du dv}}\right]+\\
    \frac{\frac{\pi}{2}-\zeta}{M}\mathrm{Tr}\left[\boldsymbol{R}_{\boldsymbol{x}}\right].
\end{equation}
Using Eq. (10) in \cite{Masouros} yields
\begin{equation}
 \frac{1}{4\delta^2}\iint\limits_{-\delta}^{\delta}{\left|\boldsymbol{a}\left(v\right)^H\boldsymbol{a}\left(u\right)\right|^2du dv}=\mathbb{E}\left[\left|\boldsymbol{a}\left(v\right)^H\boldsymbol{a}\left(u\right)\right|^2\right]=\\
 \sum_{n=0}^{M-1}{\sum_{m=0}^{M-1}{\mathrm{sinc}^2\left(2\pi\frac{d}{\lambda}\left({m}-{n}\right)\delta\right)}},
\end{equation}
which completes the proof.
\end{proof}

\begin{rem}\label{remark2}
Consider the case that $M\gg 1$. Then, from (\ref{rho_q1})
\begin{equation}\label{asym-dl}
    \mathcal{P}_{q_1}\;\stackrel{\text{(a)}}{\approx}
    \left(\zeta-1\right)\sigma_n^2+\\
    \left(\zeta-1\right)
    \left[\frac{1}{2\delta}\left(\frac{d}{\lambda}\right)^{-1}-\frac{1}{4\pi^2\delta^2}\left(\frac{d}{\lambda}\right)^{-2}f\left(\frac{d}{\lambda}\right)\right]\sum_{k=1}^{K}{p_k\beta_k}\\
    +{\left(\frac{\pi}{2}-\zeta\right)\sum_{k=1}^{K}{p_k\beta_k}},
\end{equation}
where $f\left(x\right)\triangleq \frac{2}{M}\sum_{n=1}^{M-1}{\frac{\mathrm{sin}^2\left(2\pi x\delta n\right)}{n}}$ and in ($a$) we have used Eq. (14) of \cite{Masouros}. Equation~(\ref{asym-dl}) states that, for standard one-bit quantization, increasing the spatial oversampling in a large antenna array ($d/\lambda\rightarrow 0$) increases the quantization noise power proportional to $\left(d/\lambda\right)^{-1}$.
\end{rem}

\begin{rem}
Consider the fixed-aperture case where $d_0=M\frac{d}{\lambda}$ is a constant ({\em i.e.}, the antenna spacing decreases proportionally to the increase in the number of antennas). Then, from (\ref{rho_q1})
\begin{equation}\label{asym-d0}
    \mathcal{P}_{q_1}\;{\xrightarrow[M\rightarrow \infty]{}}\;
    \left(\zeta-1\right)\left[\sigma_n^2+M\sum_{k=1}^{K}{p_k\beta_k}\right]+{\left(\frac{\pi}{2}-\zeta\right)\sum_{k=1}^{K}{p_k\beta_k}}.
\end{equation}
Equation~(\ref{asym-d0}) states that, for standard one-bit quantization, increasing the number of antennas for an array with a fixed aperture, $d_0$, increases the quantization noise power linearly with $M$. 
\end{rem}

\subsubsection{$\Sigma\Delta$ Quantization}
From (\ref{shaped-equation}), the covariance of the quantization noise for the $\Sigma\Delta$ architecture is $\boldsymbol{R}_{\boldsymbol{q}_{\Sigma\Delta}}=\boldsymbol{U}^{-1}\boldsymbol{R}_{\boldsymbol{q}}\boldsymbol{U}^{-H}$. We derive an expression for the normalized quantization noise power of the $\Sigma\Delta$ array, $\mathcal{P}_{q_{\Sigma\Delta}}$, in the next proposition. 

\begin{pro}\label{rho-qsd-pro}
The quantization noise power for spatial $\Sigma\Delta$ quantization is
\begin{equation}\label{rho_qsd}
    \mathcal{P}_{q_{\Sigma\Delta}}=\frac{2}{M}\left(\mathrm{Tr}\left[\boldsymbol{R}_{\boldsymbol{q}}\right]-\sigma^2_{q_M}\right)
    \left[1 - \mathrm{sinc}\left(2\pi\frac{d}{\lambda}\delta\right)\right]+\frac{\sigma^2_{q_M}}{M},
\end{equation}
where $\sigma_{q_M}^2=\mathbb{E}\left[|q_M|^2\right]$.
\end{pro}
\begin{proof}
Substituting $\boldsymbol{R}_{\boldsymbol{q}_{\Sigma\Delta}}=\mathbb{E}\left[\boldsymbol{U}^{-1}\boldsymbol{q}\boldsymbol{q}^H\boldsymbol{U}^{-H}\right]$ into (\ref{noise-density}) leads to 
\begin{equation}
    \mathcal{P}_{q_{\Sigma\Delta}}=\frac{1}{M}
    \frac{1}{2\delta}\int_{-\delta}^{\delta}{\mathbb{E}\left[\left|\boldsymbol{a}\left(u\right)^H\boldsymbol{U}^{-1}\boldsymbol{q}\right|^2\right]du}
    .
\end{equation}
We set $\phi=0$ due to the assumption of $u\in\left[-\delta,\delta\right]$ in the definition of the quantization noise power, and we note that
\begin{equation}
    \boldsymbol{U}^{-1}=\boldsymbol{I}_M - \boldsymbol{Z}_{-1} \; .
\end{equation}
Then
\begin{equation}
    \boldsymbol{U}^{-1}\boldsymbol{q}=\left(\boldsymbol{I}_M - \boldsymbol{Z}_{-1}\right)\boldsymbol{q}=
    \begin{bmatrix}
           q_1 \\
           q_2-q_1 \\
           \vdots \\
           q_M-q_{M-1}
    \end{bmatrix} \;.
\end{equation}

In addition, from (\ref{q-recursion}), and the fact that $\boldsymbol{R}_{\boldsymbol{qx}}\approx\boldsymbol{0}$, it can be readily shown that $\mathbb{E}\left[q_{m}q_{m\pm 1}^*\right]\approx 0$. Hence, for the sake of analysis, we approximate $\mathbb{E}\left[q_{m}q_{m'}^*\right]\approx 0,~\forall m\ne m'\in\mathcal{M}$, and therefore $\boldsymbol{R}_{\boldsymbol{q}}= \mathrm{diag}\left(\boldsymbol{p}_{\boldsymbol{q}}\right)$. Consequently,
\begin{multline}\label{integrand}
   \mathbb{E}\left[\left|\boldsymbol{a}\left(u\right)^H\boldsymbol{U}^{-1}\boldsymbol{q}\right|^2\right]=\\
    {\left|1-e^{j2\pi\frac{d}{\lambda}u}\right|^2\sum_{m=1}^{M-1}{\mathbb{E}\left[\left|q_m\right|^2\right] }+\mathbb{E}\left[\left|q_M\right|^2\right] }=
    4\left(\mathrm{Tr}\left[\boldsymbol{R}_{\boldsymbol{q}}\right]-\sigma^2_{q_M}\right)\mathrm{sin}^2\left(\pi \frac{d}{\lambda}u\right) + \sigma_{q_M}^2.
\end{multline}
By integrating (\ref{integrand}) and using some algebraic manipulation, we arrive at (\ref{rho_qsd}).
\end{proof}

\begin{rem}\label{remark4}
Consider the case that $M\gg 1$. Then, from (\ref{rho_qsd})
\begin{equation}\label{asym-ds}
    \mathcal{P}_{q_{\Sigma\Delta}}\;\stackrel{\text{(a)}}{\approx}
     \frac{4}{3}\frac{\frac{\pi}{2}-1}{2-\frac{\pi}{2}}\pi^2\delta^2\left(\frac{d}{\lambda}\right)^2p_x \; ,
\end{equation}
where in ($a$) we have used $\mathrm{sinc}\left(x\right)\approx 1 - \frac{x^2}{6}$ and
\begin{equation}
\frac{1}{M}\left(\mathrm{Tr}\left[\boldsymbol{R}_{\boldsymbol{q}}\right]-\sigma^2_{q_M}\right)\approx \frac{\frac{\pi}{2}-1}{2-\frac{\pi}{2}}p_x   
\end{equation}
for $M\gg 1$ and assuming $\boldsymbol{p}_{\boldsymbol{x}}=p_x\boldsymbol{1}$. Equation (\ref{asym-ds}) states that, by increasing the spatial oversampling ($d/\lambda\rightarrow 0$), the quantization noise power for the $\Sigma\Delta$ array tends to zero proportional to $\left(d/\lambda\right)^2$. This result is in contrast to that for the standard one-bit quantization power, which was shown earlier to increase proportional to $\left(d/\lambda\right)^{-1}$. Hence, the spatial $\Sigma\Delta$ architecture brings about an oversampling gain of $\left(d/\lambda\right)^{3}$ compared to the standard one-bit architecture. While this is a promising result, as mentioned earlier the practical limitations of placing antenna elements close to each other prevent us from achieving a high degree of spatial oversampling.
\end{rem}

\begin{rem}
Consider the case that $d_0=M\frac{d}{\lambda}$ is a constant. Then, from (\ref{rho_qsd})
\begin{equation}\label{asym-d0sd}
    M^2\mathcal{P}_{q_{\Sigma\Delta}}\;{\xrightarrow[M\rightarrow \infty]{}}\;\frac{4}{3}\frac{\frac{\pi}{2}-1}{2-\frac{\pi}{2}}\pi^2\delta^2{d_0^2}p_x.
\end{equation}
Equation~(\ref{asym-d0sd}) states that, for spatial $\Sigma\Delta$ quantization, increasing the number of antennas for an array with a fixed aperture, $d_0$, decreases the quantization noise power proportional to ${1}/{M^2}$. Hence, the spatial $\Sigma\Delta$ architecture brings about an oversampling gain of $M^{3}$ compared to the standard one-bit architecture. 
\end{rem}

In the next section, we study the spectral efficiency of a massive MIMO system with spatial $\Sigma\Delta$ processing and discuss the impact of the spatial $\Sigma\Delta$ architecture on the system performance. 

\section{Spectral Efficiency}\label{sec:spectral efficiency}
In this section, we study the SE of a massive MIMO system with spatial $\Sigma\Delta$ processing. We consider maximum ratio combining (MRC) and zero-forcing (ZF) receivers. We derive here an approximation for the SE of the system with an MRC receiver, and evaluate the SE for the ZF receiver in the next section, numerically. We first present the case where perfect channel state information (CSI) is assumed to be available at the BS, and then we discuss the impact of imperfect CSI on the system performance at the end of the section. 

From (\ref{system-model}) and (\ref{shaped-equation}), the received signal at a BS with a $\Sigma\Delta$ architecture can be modeled as
\begin{equation}\label{shaped-received}
    \boldsymbol{y}=\boldsymbol{G}\boldsymbol{P}^{\frac{1}{2}}
  \boldsymbol{{s}} + {\boldsymbol{n}}+\boldsymbol{U}^{-1}\boldsymbol{q}.
\end{equation}
Denoting the linear receiver by $\boldsymbol{W}$, we have
\begin{equation}
    \hat{\boldsymbol{s}}=\boldsymbol{W}^H\boldsymbol{G}\boldsymbol{P}^{\frac{1}{2}}
  \boldsymbol{{s}} + \boldsymbol{W}^H{\boldsymbol{n}}+\boldsymbol{W^H}\boldsymbol{U}^{-1}\boldsymbol{q} \; ,
\end{equation}
and the $k$th element of $\hat{{\boldsymbol{s}}}$ is given by
\begin{equation}\label{kth user received signal}
  \hat{s}_k=\sqrt{p_k}{{\boldsymbol{w}}}_{k}^H{{\boldsymbol{g}}}_{k}s_k
  +
  \sum_{i=1,i\ne k}^{K}{\sqrt{p_k}{{{\boldsymbol{w}}}_{k}^H}{{\boldsymbol{g}}}_{i}s_i}+\\
  {\boldsymbol{w}}_{k}^H\boldsymbol{n}
  +\boldsymbol{w}_{k}^{H}\boldsymbol{U}^{-1}\boldsymbol{q},
\end{equation}
where $\boldsymbol{w}_k$ is the $k$th column of ${\boldsymbol{W}}$. We assume the BS treats ${{\boldsymbol{w}}}_{k}^H{{\boldsymbol{g}}}_{k}$ as the desired channel and the other terms of (\ref{kth user received signal}) as worst-case Gaussian noise when decoding the signal.
Consequently, a lower bound for the ergodic achievable SE at the $k$th user can be written as \cite{Matthaiou}
\begin{equation}\label{spec_eff}
	\mathcal{S}_k=\mathbb{E}\left[\mathrm{log}_2\left(1+\frac{p_k\left|\boldsymbol{w}_{k}^H\boldsymbol{g}_{k}\right|^2}{\Omega}\right)\right],
\end{equation}
where
\begin{equation}\label{intef_n_q}
\Omega=\\
  \sum_{i=1,i\ne k}^{K}{p_k\left|{{{\boldsymbol{w}}}_{k}^H}{{\boldsymbol{g}}}_{i}\right|^2}+\sigma_n^2\|\boldsymbol{w}_k\|^2+\boldsymbol{w}_{k}^{H}\boldsymbol{U}^{-1}\boldsymbol{R}_{\boldsymbol{q}}\boldsymbol{U}^{-H}\boldsymbol{w}_k.
\end{equation}

\subsection{MRC Receiver}\label{MRC}
For the case of an MRC receiver, $\boldsymbol{W}=\boldsymbol{G}$. The following proposition presents an approximation for the achievable SE of a massive MIMO system with spatial $\Sigma\Delta$ processing and an MRC receiver.

\begin{pro}\label{Pro3}
For a massive MIMO system employing a spatial $\Sigma\Delta$ architecture and an MRC receiver, the SE of the $k$th user assuming perfect CSI is given by eq.~(\ref{approx_theo}), where $\boldsymbol{\Sigma}_{ik}\triangleq\frac{1}{L}\boldsymbol{A}_i^H\boldsymbol{A}_k$.

\small
\begin{multline}\label{approx_theo}
\mathcal{S}_k\approx
\mathrm{log}_2\left(1+\frac{p_k\beta_k\left(\left|\mathrm{Tr}\left[\boldsymbol{\Sigma}_{kk}\right]\right| ^2 + \mathrm{Tr}\left[\boldsymbol{\Sigma}_{kk}^{2}\right]\right)}{\sum_{i=1,i\ne k}^{K}{p_i\beta_i\mathrm{Tr}\left[\boldsymbol{\Sigma}_{ik}\boldsymbol{\Sigma}_{ik}^H\right]}+\sigma_n^2\mathrm{Tr}\left[\boldsymbol{\Sigma}_{kk}\right]+\frac{4}{L}\left(\mathrm{Tr}\left[\boldsymbol{R}_{\boldsymbol{q}}\right]-\sigma_{q_M}^2\right)\sum_{\ell=1}^{L}{\mathrm{sin}^2\left(\frac{\phi-2\pi\frac{d}{\lambda}\mathrm{sin\left(\theta_{k\ell}\right)}}{2}\right)}+\sigma_{q_M}^2}\right)	
\end{multline}

\end{pro}
\begin{proof}
From \cite{Matthaiou}, an approximation for (\ref{spec_eff}) can be calculated as
\begin{equation}\label{se-appr}
	\mathcal{S}_k\approx\mathrm{log}_2\left(1+\frac{p_k\mathbb{E}\left[\left|\boldsymbol{w}_{k}^H\boldsymbol{g}_{k}\right|^2\right]}{\mathbb{E}\left[\Omega\right]}\right).
\end{equation}
By setting $\boldsymbol{w}_k=\boldsymbol{g}_k$ and using Lemma 2 of \cite{Emil2} and Lemma 1 of \cite{Hessam_SPAWC}, the expected values of the desired signal, interference, and thermal noise can be readily calculated. 
For the quantization noise term, note that 
\begin{equation}
    \boldsymbol{U}^{-1}=\boldsymbol{I}_M - e^{-j\phi}\boldsymbol{Z}_{-1}.
\end{equation}
Therefore,
\begin{equation}
    \boldsymbol{U}^{-1}\boldsymbol{q}=\left(\boldsymbol{I}_M - e^{-j\phi}\boldsymbol{Z}_{-1}\right)\boldsymbol{q}=
    \begin{bmatrix}
           q_1 \\
           q_2-e^{-j\phi}q_1 \\
           \vdots \\
           q_M-e^{-j\phi}q_{M-1}
    \end{bmatrix}.
\end{equation}
In addition, the $k$th user channel vector can be written as
\begin{equation}
    \boldsymbol{g}_k=\sqrt{\frac{\beta_k}{L}}\sum_{l=1}^{L}{h_{kl}\boldsymbol{a}\left(\theta_{kl}\right)},
\end{equation}
where $h_{kl}$ is the $l$th element of $\boldsymbol{h}_k$. Hence, 
\begin{equation}
   \mathbb{E}\left[\left|\boldsymbol{g}_k^H\boldsymbol{U}^{-1}\boldsymbol{q}\right|^2\right]=\\
    \frac{\beta_k}{L}\mathbb{E}\left[\left|\sum_{\ell=1}^{L}{h_{kl}\left(1-e^{-j\phi}z_{kl}\right)\sum_{m=1}^{M-1}{q_m z_{kl}^{m-1}}+q_M z_{kl}^{M-1}}\right|^2\right],
\end{equation}
which, after some algebraic manipulation, leads to (\ref{approx_theo}) and the proof is complete.
\end{proof}

\begin{rem}\label{se-remark}
The noise shaping characteristic of the spatial $\Sigma\Delta$ architecture is explicitly manifested in~(\ref{approx_theo}). A similar characteristic is observed in \cite{ShaoMLS19} for $\Sigma\Delta$ precoding. It shows the importance of the design parameter $\phi$ which should be chosen to minimize $\mathcal{G}=\frac{1}{L}\sum_{\ell=1}^{L}{\mathrm{sin}^2\left(\frac{\phi-2\pi\frac{d}{\lambda}\mathrm{sin\left(\theta_{k\ell}\right)}}{2}\right)}$ for all users. By writing the steering angle as $\phi=2\pi\frac{d}{\lambda}\mathrm{sin\left(\theta\right)}$, we have
\begin{equation}\label{d-impact}
\mathcal{G}=\frac{1}{L}\sum_{\ell=1}^{L}{\mathrm{sin}^2\left(\pi\frac{d}{\lambda}\Bigl(\mathrm{sin}\left(\theta\right)-\mathrm{sin}\left(\theta_{k\ell}\right)\Bigr)\right)}. 
\end{equation}
Eq.~(\ref{d-impact}) indicates that $\mathcal{G}$ could be made arbitrarily small by decreasing the relative antenna spacing $d/\lambda$ (the {spatial oversampling gain}) or $\mathrm{sin}\left(\theta\right)-\mathrm{sin}\left(\theta_{k\ell}\right)$ (the {angle steering gain}). However, physical constraints on the antenna spacing and larger angular spreads, $\Theta$, limit the lower bound on $\mathcal{G}$. {For the case that $L\gg 1$, $\mathrm{sin}\left(\theta_{k\ell}\right)=u_{k\ell}$ can be taken as a random variable uniformly distributed in $\left[\delta_1,\delta_2\right]$ where $\delta_1=\mathrm{sin}\left(\theta_0-\frac{\Theta}{2}\right)$ and $\delta_2=\mathrm{sin}\left(\theta_0+\frac{\Theta}{2}\right)$. Hence,
\begin{equation}\label{Gphi}
    \mathcal{G}\approx\frac{1}{\delta_2-\delta_1}\int_{\delta_1}^{\delta_2}{\mathrm{sin}^2\left(\frac{\phi-2\pi\frac{d}{\lambda}u}{2}\right)du}=\\
    \frac{1}{2}+\frac{1}{4\pi}\left(\frac{d}{\lambda}\right)^{-1}\frac{1}{\delta_2-\delta_1}\left(b_0\mathrm{sin}\left(\phi\right)-b_1\mathrm{cos}\left(\phi\right)\right),
\end{equation}
where
\begin{equation*}
    b_0=\mathrm{cos}\left(2\pi\frac{d}{\lambda}\delta_2\right)-\mathrm{cos}\left(2\pi\frac{d}{\lambda}\delta_1\right)
\end{equation*}
\begin{equation*}
    b_1=\mathrm{sin}\left(2\pi\frac{d}{\lambda}\delta_2\right)-\mathrm{sin}\left(2\pi\frac{d}{\lambda}\delta_1\right).
\end{equation*}
In this case, the optimal value of the steering angle that minimizes $\mathcal{G}$ can be simply derived as 
\begin{equation} \label{phi_star}
\phi^{\star}=\begin{cases} 0 & \delta_2=-\delta_1 \\ -\mathrm{tan}^{-1}\left(\frac{b_0}{b_1}\right) & \text{otherwise}\end{cases}
\end{equation} 
which indicates that the optimal steering angle is dependent on $\delta_1$, $\delta_2$, and the relative antenna spacing $d/\lambda$.}
\end{rem}

\subsection{ZF receiver}\label{ZF-receiver}
For the ZF receiver, $\boldsymbol{W}=\boldsymbol{G}\left(\boldsymbol{G}^H\boldsymbol{G}\right)^{-1}$. After substituting this for $\boldsymbol{W}$ in~(\ref{se-appr}), the SE achieved for the $k$th user with the $\Sigma\Delta$ architecture and ZF receiver can be written as in~(\ref{ZF_SE}) at the top of the next page. Although~(\ref{ZF_SE}) does not provide direct insight into the effect of the shaped quantization noise on the SE, in Section~\ref{sec:Simulation} we numerically evaluate this expression and show the superior performance of the $\Sigma\Delta$ architecture compared with standard one-bit quantization.  
\begin{figure*}[!t]
\begin{equation}\label{ZF_SE}
\mathcal{S}_k=\mathbb{E}\left[\mathrm{log}_2\left(1+\frac{p_k}{\|\boldsymbol{w}_k\|^2\sigma_n^2+\left[\left(\boldsymbol{G}^H\boldsymbol{G}\right)^{-1}\boldsymbol{G}^H\boldsymbol{U}^{-1}\boldsymbol{R}_{\boldsymbol{q}}\boldsymbol{U}^{-H}\boldsymbol{G}\left(\boldsymbol{G}^{H}\boldsymbol{G}\right)^{-1}\right]_{kk}}\right)\right]
\end{equation}
\hrulefill
\vspace*{4pt}
\end{figure*}

\section{Numerical Results}\label{sec:Simulation}
In this section, we numerically evaluate the SE performance of $\Sigma\Delta$ massive MIMO systems in various scenarios. We assume static-aware power control in the network \cite{Emil Bj} so that $p_k=p_0/\beta_k$. In all of the cases considered, unless otherwise noted, we assume $M=100$ antennas, $K=10$ users, and an angular spread of $\Theta=40^{\circ}$ centered at $\theta_0=30^{\circ}$. We assume the same DoAs for all users, i.e., $\boldsymbol{A}_k=\boldsymbol{A},~\forall k$, drawn uniformly from the interval $\left[10^{\circ},50^{\circ}\right]$, which corresponds to $u=\sin(\theta)\in [0.17, 0.77]$, 
and the steering angle of the $\Sigma\Delta$ array is set to $\phi=2\pi\frac{d}{\lambda}\mathrm{sin}\left(\theta_0\right)$. The SNR is defined to be $\mathrm{SNR}\triangleq\frac{p_0}{\sigma_n^2}$. We further assume CSCG symbols and $10^4$ Monte Carlo trials for the simulations.

\begin{figure}
\centering
\includegraphics[width=1\textwidth]
{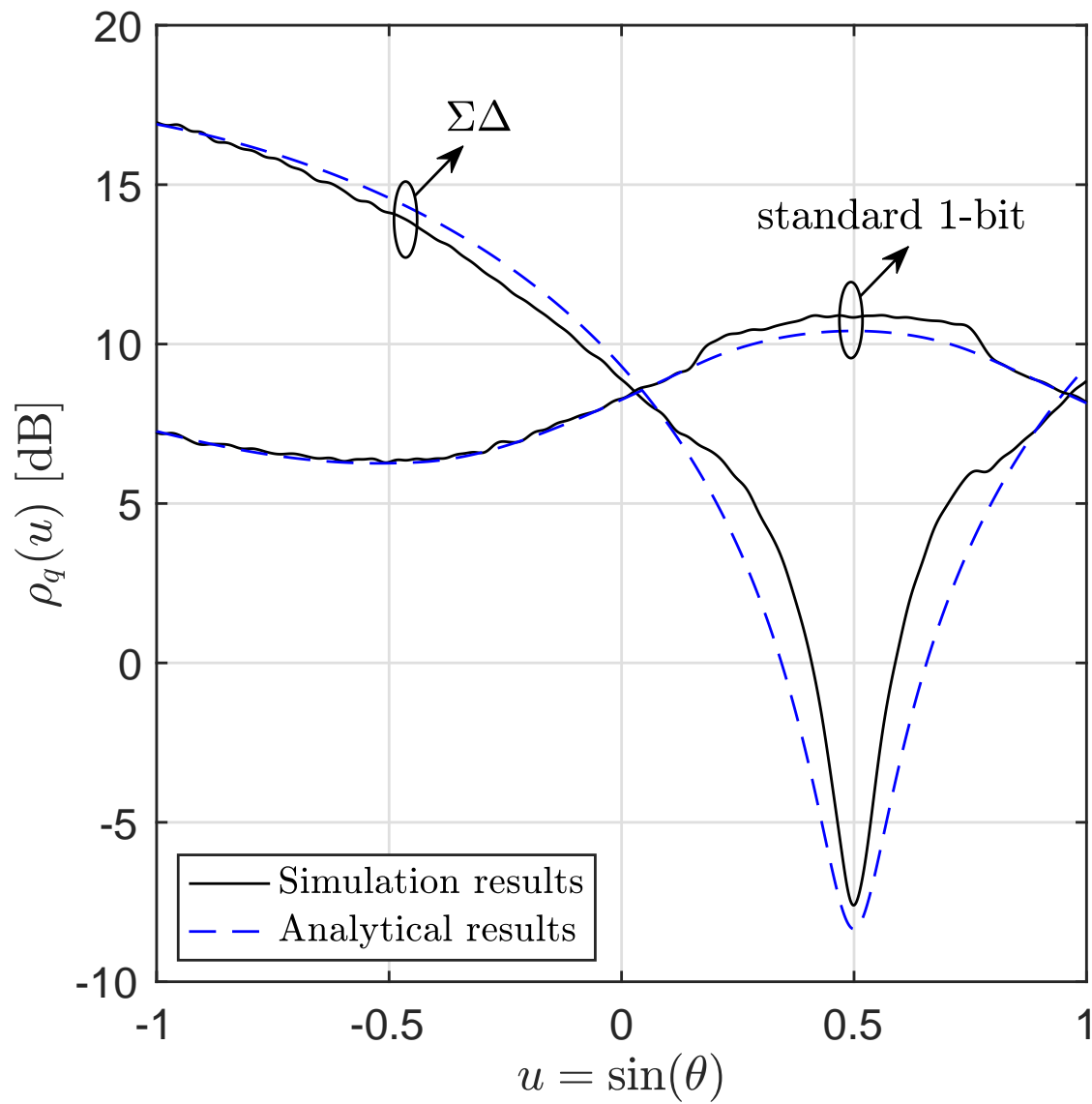}
\caption{Spatial spectrum of the quantization noise for the $\Sigma\Delta$ and standard one-bit architectures when $L=50$, $d=\lambda/4$, and $\mathrm{SNR}=0~\mathrm{dB}$.}
\label{fig1}
\end{figure}
Fig. \ref{fig1} shows the simulated and analytically derived quantization noise power density, i.e., $\rho_q\left(u\right),~u\in\left[-1,1\right]$, for $\Sigma\Delta$ and standard one-bit quantization when the relative antenna spacing is $d=\lambda/4$. We see that the quantization noise power for the $\Sigma\Delta$ array is substantially lower over the angles where the users are present, while the effect is the opposite for standard one-bit quantization -- the quantization noise is higher for angles where the amplitude of the received signals is larger. We also observe that there is excellent agreement between the simulations and our theoretically derived expressions for both cases. Note that the careful design of the quantizer output levels is a critical aspect for achieving the desired $\Sigma\Delta$ noise shaping characteristic shown here.

\begin{figure}
\centering
\includegraphics[width=1\textwidth]
{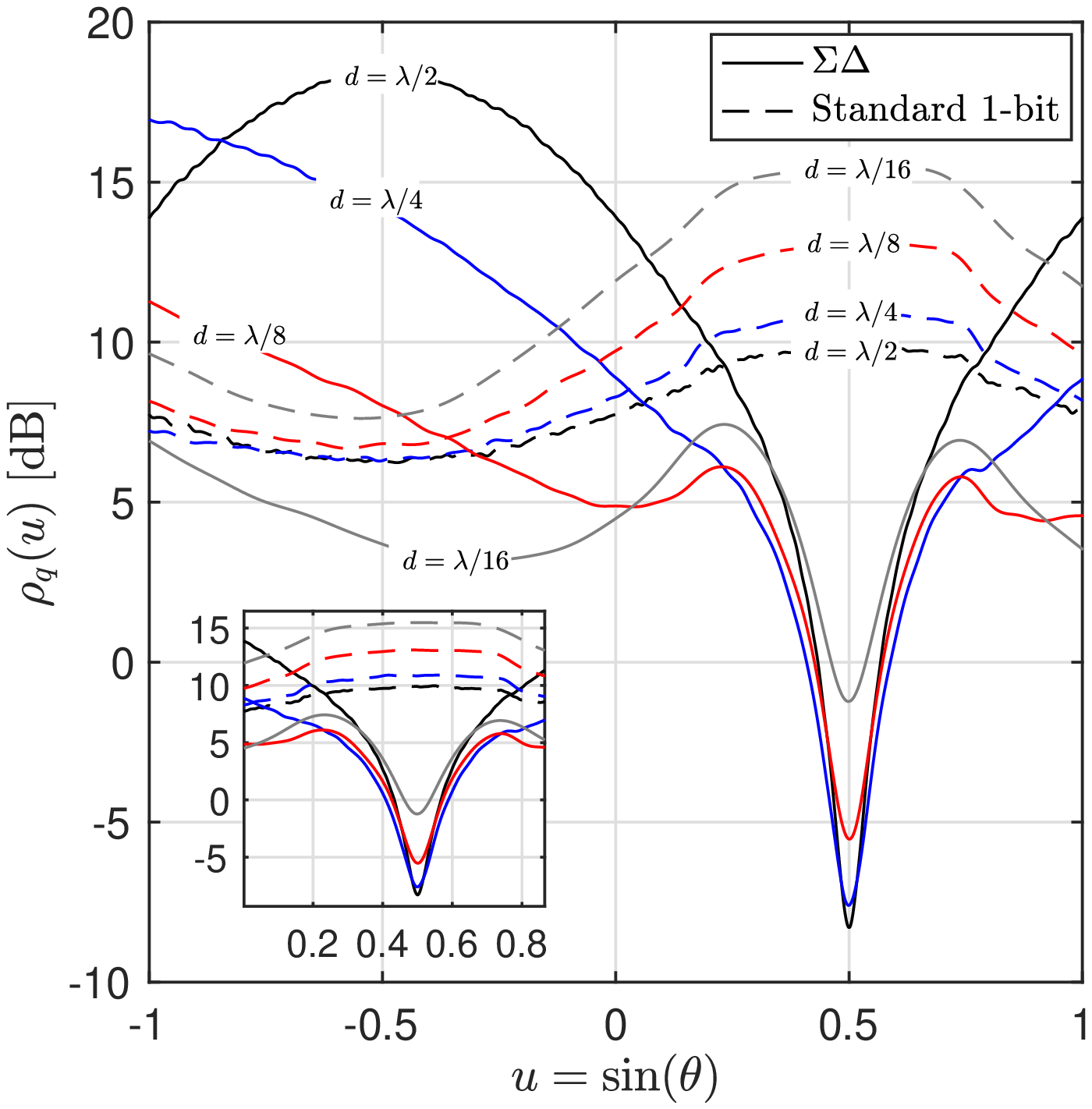}
\caption{Spatial spectrum of the quantization noise for the $\Sigma\Delta$ and standard one-bit architectures for different antenna spacings when $L=50$ and $\mathrm{SNR}=0~\mathrm{dB}$.}
\label{fig2}
\end{figure}
The impact of spatial oversampling on the shape of the quantization noise spectrum is illustrated in Fig.~\ref{fig2}. {We see from the figure that, as discussed in Remarks~\ref{remark2}, the quantization noise power for the standard one-bit ADC architecture grows as $d/\lambda$ decreases.}
Analogously to temporal $\Sigma\Delta$ modulation where increasing the sampling rate helps to push the quantization noise to higher frequencies and widen the quant\-iza\-tion-noise-free band, we can reduce the quantization noise power over wider angular regions by placing the antenna elements of the array closer together. For example, when $d=\lambda/2$, the $\Sigma\Delta$ quantization noise power is below that of the standard one-bit quantizer over a beamwidth of $40^{\circ}$. This beamwidth increases to about $80^{\circ}$, $150^{\circ}$, and $180^{\circ}$ for $d=\lambda/4$, $d=\lambda/8$, and $d=\lambda/16$, respectively. Mutual coupling will impact these results as $d$ decreases, but both the standard one-bit and $\Sigma\Delta$ approaches would be expected to degrade.

\begin{figure}
\centering
\includegraphics[width=1\textwidth]
{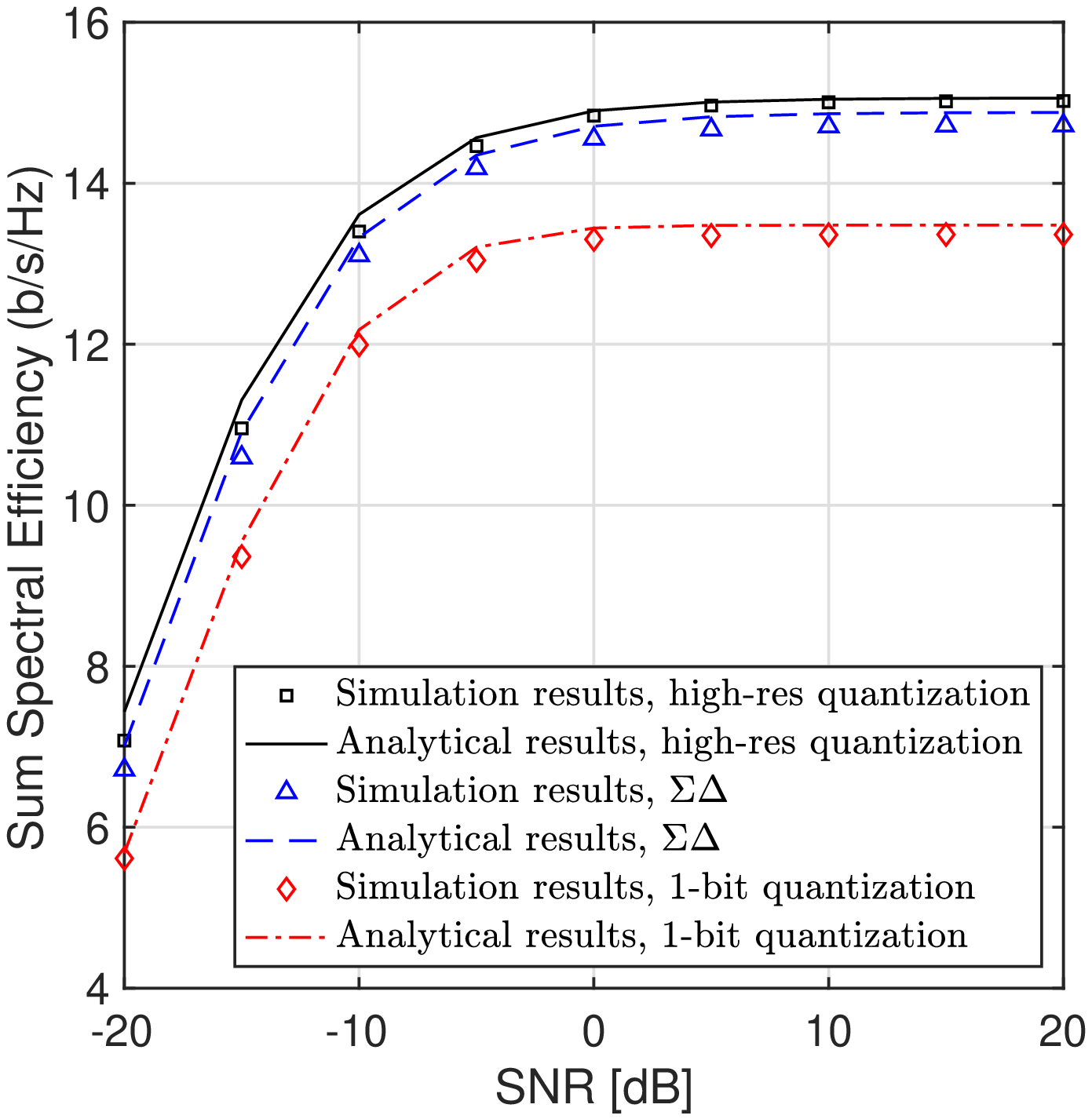}
\caption{SE versus SNR for MRC receiver with perfect CSI, $L=50$, and $d=\lambda/4$.}
\label{fig3}
\end{figure}

In Fig. \ref{fig3}, we compare the SE performance of $\Sigma\Delta$ and standard one-bit quantization for the case of an MRC receiver. It is clear that the derived theoretical SE expression in (\ref{approx_theo}) very closely matches the simulated value of the expression in (\ref{spec_eff}). The one-bit $\Sigma\Delta$ implementation achieves a significantly increased SE compared with standard one-bit quantization, and performs nearly identically to an MRC receiver with infinite resolution ADCs. It should be emphasized that this performance gain is achieved without paying a significant penalty in terms of power consumption (as with mixed-ADC architectures) or complicated processing (as required by non-linear receivers).

\begin{figure}
\centering
\includegraphics[width=1\textwidth]
{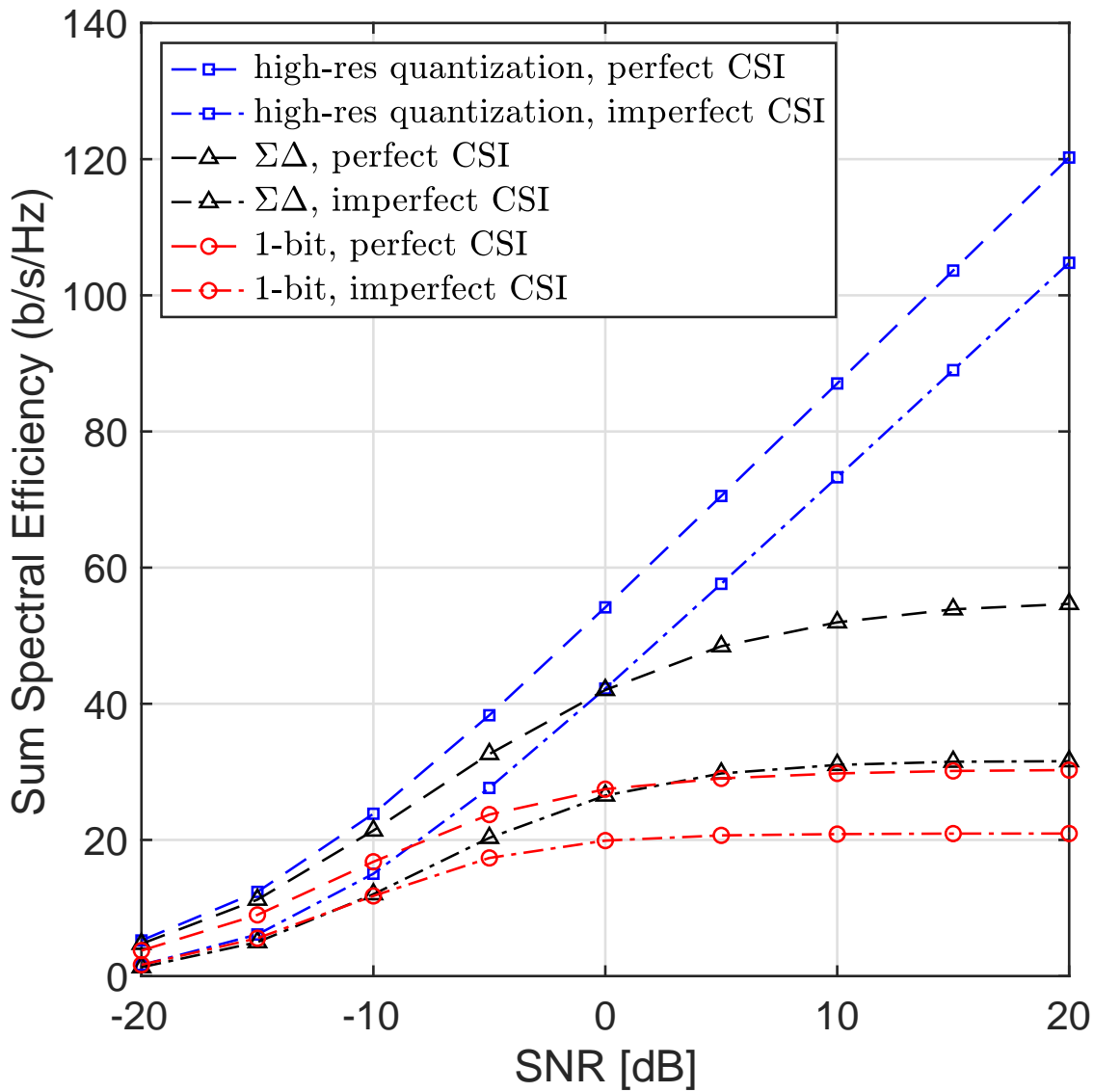}
\caption{SE versus SNR for ZF receiver with and without channel estimation error. $L=20$, $d=\lambda/4$.}
\label{fig4}
\end{figure}

In Fig. \ref{fig4} and \ref{fig5}, we numerically evaluate the SE when the ZF receiver is employed, using Eq.~(\ref{ZF_SE}). The SE improvement of $\Sigma\Delta$ processing is much greater than for the case of MRC. For example, at $\mathrm{SNR}=0~\mathrm{dB}$, about a $50\%$ improvement in SE can be achieved by the spatial $\Sigma\Delta$ architecture compared with standard one-bit quantization, which confirms its ability to provide high SE with a simple architecture and low power consumption.

The effect of channel estimation error on the performance of the algorithms is also shown in Fig. \ref{fig4} for the ZF receiver. For these results, we used a least squares (LS) channel estimator for each of the algorithms. In this approach, the channel estimate, $\hat{\boldsymbol{G}}$, becomes
\begin{equation}\label{CE}
    \hat{\boldsymbol{G}}=\frac{1}{\eta\sqrt{ p_0}}\boldsymbol{P}_{\boldsymbol{A}}\boldsymbol{Y}\boldsymbol{\Phi}^{*},
\end{equation}
where $\eta$ is the training length, $\boldsymbol{P}_{\boldsymbol{A}}=\boldsymbol{A}\boldsymbol{A}^{\dagger}$ is the orthogonal projection onto $\boldsymbol{A}$, $\boldsymbol{Y}\in\mathbb{C}^{M\times\eta}$ is the received data during the channel estimation phase, and $\boldsymbol{\Phi}\in\mathbb{C}^{\eta\times K}$ is the orthogonal pilot matrix satisfying $\boldsymbol{\Phi}^H\boldsymbol{\Phi}=\eta\boldsymbol{I}$. We set $\eta=K$ and choose $\boldsymbol{\Phi}$ from among the columns of the discrete Fourier transform (DFT) matrix. Note that for the case of high-resolution quantization, $\boldsymbol{Y}=\sqrt{p_0}\boldsymbol{G}\boldsymbol{\Phi}^T+\boldsymbol{N}$, where the elements of $\boldsymbol{N}$ are independent $\mathcal{CN}\left(0,1\right)$ random variables. For standard one-bit and $\Sigma\Delta$ quantization, we pass $\boldsymbol{Y}$ through the corresponding quantization, and plug the output into (\ref{CE}) for channel estimation.
Fig. \ref{fig5} shows the performance of the ZF receiver with and without perfect CSI versus the number of antennas. The presence of imperfect CSI obviously degrades all of the algorithms, but we see that the $\Sigma\Delta$ architecture provides a way to successfully bridge the performance gap between standard one-bit and high-resolution quantization with only a minimal increase in hardware complexity.

\section{Conclusion}\label{sec:concolusion}
In this paper, we studied the performance of massive MIMO systems employing spatial one-bit $\Sigma\Delta$ quantization. Using an element-wise Bussgang approach, we derived an equivalent linear model in order to analytically characterize the spectral efficiency of a massive MIMO base station with a $\Sigma\Delta$ array, and we compared the results with the performance achieved by an array that employs standard one-bit quantization. Our results demonstrated that the spatial $\Sigma\Delta$ architecture can scale down the quantization noise power proportional to the square of the spatial oversampling rate. This can be interpreted as scaling down the quantization noise power proportional to the inverse square of the number of antennas at the BS for space-constrained arrays. This result gains more importance by noting that in standard one-bit quantization, the quantization noise power grows proportional to the inverse of the spatial oversampling rate, or equivalently, proportional to the number of antennas at the BS in space-constrained arrays. Furthermore, it was shown how this capability allows the spatial $\Sigma\Delta$ architecture to bridge the SE gap between infinite resolution and standard one-bit quantized systems. For the ZF receiver, the spatial $\Sigma\Delta$ architecture can outperform standard one-bit quantization by about $50\%$, and achieve almost the same performance as an infinite resolution system for the MRC receiver. While these results were obtained by assuming the availability of perfect CSI at the BS, we also showed that the spatial $\Sigma\Delta$ architecture is able to alleviate the adverse impact of quantization noise in the presence of channel estimation error.

\begin{figure}
\centering
\includegraphics[width=1\textwidth]
{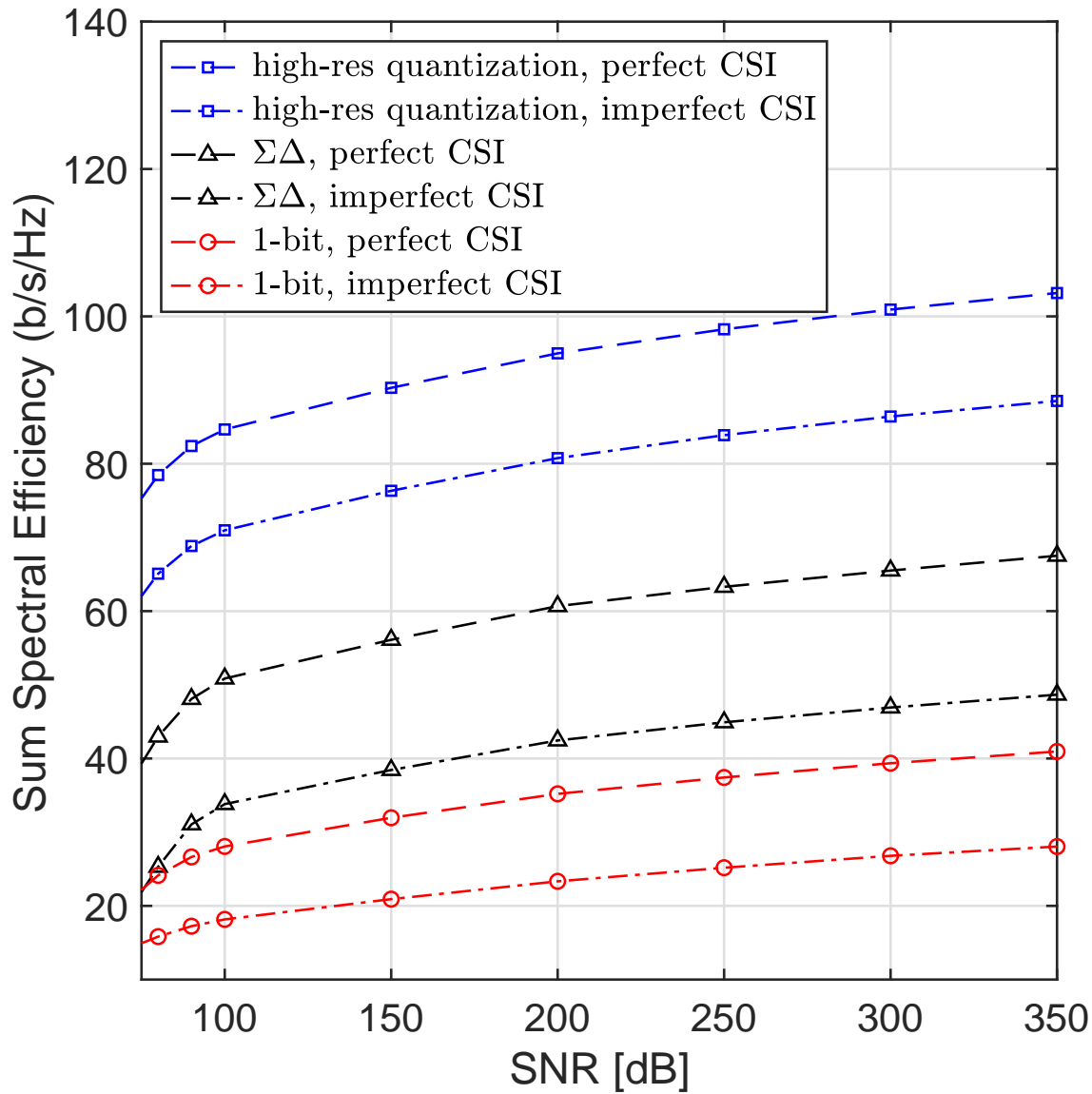}
\caption{SE versus $M$ for ZF receiver with and without channel estimation error. $L=15$, $d=\lambda/4$, $\mathrm{SNR}=10$ dB.}
\label{fig5}
\end{figure}

\end{document}